\documentclass[letterpaper, 10 pt, conference]{ieeeconf}  

\IEEEoverridecommandlockouts                              

\usepackage{amsmath} 
\usepackage{amssymb}  
\usepackage{mathtools}
\usepackage[disable]{todonotes}

\usepackage[normalem]{ulem}

\usepackage{graphicx}      
\usepackage{subcaption}    
\usepackage{float}         
\usepackage{acronym}       
\let\labelindent\relax
\usepackage{enumitem}
\usepackage{cite}
\usepackage{pgfplots}
\pgfplotsset{compat=1.18} 

\acrodef{pbh}[PBH]{Popov-Belevitch-Hautus} 
\acrodef{uavs}[UAVs]{unmanned aerial vehicles}
\acrodef{mass}[MASs]{multi-agent systems}
\acrodef{lti}[LTI]{linear time-invariant}
\acrodef{uav}[UAV]{unmanned aerial vehicle}
\acrodef{idr}[IDR]{\emph{infinitesimal distance rigidity}}
\acrodef{rbm}[RBM]{\emph{rigid body motion}}
\acrodef{rbms}[RBMs]{\emph{rigid body motions}}
\acrodef{fvt}[FVT]{\emph{final value theorem}}
\acrodef{svd}[SVD]{\emph{singular value decomposition}} 
\acrodef{evd}[EVD]{\emph{eigen value decomposition}} 
\acrodef{siso}[SISO]{\emph{single-input single-output}}
\acrodef{mimo}[MIMO]{\emph{multi-input multi-output}} 
\acrodef{odes}[ODEs]{\emph{ordinary differential equations}} 
\acrodef{psd}[PSD]{\emph{positive semi-definite}} 
\acrodef{nsd}[NSD]{\emph{negative semi-definite}} 

\newtheorem{theorem}{Theorem}

\newtheorem{proposition}{Proposition}
\newtheorem{corollary}{Corollary}
\newtheorem{definition}{Definition}
\newtheorem{remark}{Remark}

\title{\LARGE \bf
The Geometry of Transmission Zeros in Distance-Based Formations*
}

\author{Solomon Goldgraber Casspi and Daniel Zelazo, \IEEEmembership{Senior Member, IEEE}
\thanks{*This work was supported by the Israel Science Foundation grant
no. 453/24 and the Gordon Center for Systems Engineering.}
\thanks{Both authors are with the Stephen B. Klein Faculty of Aerospace Engineering, Technion-Israel Institute of Technology, Haifa 320003, Israel 
        {\tt\small solomon.g@campus.technion.ac.il, dzelazo@technion.ac.il}}%
}

\begin{document}

\maketitle
\thispagestyle{empty}
\pagestyle{empty}

\begin{abstract}
This letter presents a geometric input-output analysis of distance-based formation control, focusing on the phenomenon of steady-state signal blocking between actuator and sensor pairs. We characterize steady-state multivariable transmission zeros, where fully excited rigid-body and deformational modes destructively interfere at the measured output. By analyzing the DC gain transfer matrix of the linearized closed-loop dynamics, we prove that for connected, flexible frameworks, structural transmission zeros are strictly non-generic; the configuration-dependent cross-coupling required to induce them occupies a proper algebraic set of measure zero. However, because extracting actionable sensor-placement rules from these complex algebraic varieties is analytically intractable, we restrict our focus to infinitesimally rigid formations. For these baselines, we prove that the absence of internal flexes forces the zero-transmission condition to collapse into an explicit affine hyperplane defined by the actuator and the global formation geometry, which we term the spatial locus of transmission zeros. Finally, we introduce the global transmission polygon--a convex polytope constructed from the intersection of these loci. This construct provides a direct geometric synthesis rule for robust sensor allocation, guaranteeing full-rank steady-state transmission against arbitrary single-node excitations.
\end{abstract}

\section{Introduction}
The coordination of autonomous multi-agent systems relies heavily on distributed control protocols, where agents maintain desired formations using only local relative measurements \cite{akyildiz2002survey, survey_general_2018}. Distance-based formation control, driven by gradient-descent laws derived from edge-tension potentials, has emerged as a robust solution for operations in GPS-denied or decentralized environments \cite{krick2009stabilisation, survey_formation_control_2015_automatica, babazadeh2019distance, OLFATISABER2002495}. 

However, under a localized exogenous disturbance--or when control authority is restricted to a single ``leader'' agent for maneuvering--the inherent spatial geometry may completely block these inputs from reaching measured outputs. From a multivariable control perspective, this phenomenon occurs when the steady-state input-output mapping becomes rank-deficient. In such cases, perfect destructive interference causes certain disturbance directions to vanish entirely at the measured output, which is formally classified as a steady-state transmission zero \cite[Sec. 4.5]{skogestad2005multivariable}.

While the algebraic and graph-theoretic conditions for transmission zeros have been explored in standard integrator consensus networks--linking signal blocking to shortest paths and residual graph structures \cite{briegel2011zeros, abad_torres_2015}--the specific geometric complexities of distance-based formations remain entirely uncharacterized. Recent geometric work on distance-based formations has primarily focused on state-space uncontrollability \cite{casspi2025geometry}, where an actuator fails to excite a hidden mode because it physically lies exactly at the center of a pure rotation.

In this letter, we address a fundamentally different phenomenon: if the actuator and sensor are not pinned \cite{casspi2025geometry}, meaning all relevant modes are successfully excited and measured, what exact spatial configurations induce a steady-state transmission zero? Translating these formal algebraic rank-deficient conditions into explicit spatial coordinates is critical for ensuring robust network operation. 

By analyzing the linearized formation dynamics, the main contribution of this letter is to geometrically characterize these steady-state transmission zeros. Specifically, our contributions are threefold:
\begin{itemize}
    \item We derive a closed-form algebraic condition for when the steady-state DC gain matrix becomes rank-deficient. Using this, we prove that transmission zeros in connected, flexible frameworks are strictly non-generic, occurring exclusively on an algebraic set of measure zero.
    \item For infinitesimally rigid formations, we define and characterize the \textit{spatial locus of transmission zeros}. We prove that the linearized steady-state signal perfectly vanishes if and only if the sensor is placed on a specific line (an affine hyperplane) dictated entirely by the actuator's position and the formation's global geometry. 
    \item We introduce the \textit{global transmission polygon}, a safe convex region constructed from the intersection of valid half-planes. This construct provides a direct geometric rule for placing sensors to guarantee full rank of the DC gain matrix. 
\end{itemize}

The remainder of this letter is organized as follows. Section \ref{sec:preliminaries} introduces the preliminaries of graph rigidity and formulates the linearized input-output network dynamics. Section \ref{sec:steady_state_analysis} derives the steady-state DC gain rank-loss conditions, proves that transmission zeros in connected, flexible networks are constrained to measure-zero algebraic sets, and extracts explicit geometric boundaries--the spatial locus and global transmission polygon--for infinitesimally rigid formations. Section \ref{sec:simulations} presents numerical simulations validating these exact geometric constraints, and Section \ref{sec:conclusion} provides concluding remarks.

\paragraph*{Notations}
The sets of real and complex numbers are $\mathbb{R}$ and $\mathbb{C}$. The $n \times n$ identity matrix is $I_n$, and $e_i \in \mathbb{R}^n$ is the $i$-th standard basis vector. The Kronecker product is $\otimes$. For a block vector $v \in \mathbb{R}^{2n}$, its $2$-dimensional spatial component at agent $i$ is $[v]_i := (e_i^\top \otimes I_2)v \in \mathbb{R}^2$. The generator of 2D rotations is $\Omega = \begin{bsmallmatrix} 0 & -1 \\ 1 & 0 \end{bsmallmatrix}$. The null space and range space of a matrix $A$ are $\ker(A)$ and $\operatorname{Im}(A)$. Finally, for a coordinate vector $p \in \mathbb{R}^{2n}$, $\mathbb{R}[p]$ denotes the ring of polynomials in $p$ with real coefficients, and $\mathbb{R}(p)$ denotes the field of rational functions in $p$.

\section{Preliminaries and Problem Formulation} \label{sec:preliminaries}
This section establishes the mathematical foundation required to formalize the geometric transmission zero problem. First, Subsection \ref{subsec:graph-rigidity} models the spatial embedding of the network using graph rigidity theory. Next, Subsection \ref{subsec:IO-dynamics} formulates the input-output formation control problem, culminating in the core objective of identifying the geometric configurations that induce a steady-state transmission zero.

\subsection{Graph Rigidity and Modal Subspaces}
\label{subsec:graph-rigidity}
Consider a formation of $n$ agents in $\mathbb{R}^2$ with an undirected sensing graph $\mathcal{G} = (\mathcal{V}, \mathcal{E})$, where $\mathcal{V} = \{1, \dots, n\}$ and $\mathcal{E} \subseteq \mathcal{V} \times \mathcal{V}$. Each agent obeys single-integrator dynamics $\dot{p}_i = u_i$, with position $p_i \in \mathbb{R}^2$ and control input $u_i \in \mathbb{R}^2$. The stacked configuration is $p = [p_1^\top, \dots, p_n^\top]^\top \in \mathbb{R}^{2n}$.

The formation geometry is encoded by the rigidity function $F : \mathbb{R}^{2n} \to \mathbb{R}^{|\mathcal{E}|}$, where each component evaluates the squared inter-agent distance, $F_\ell(p) = \frac{1}{2}\|p_i - p_j\|^2$, for edge $\ell = (i,j) \in \mathcal{E}$ \cite{rigidity_of_graphs_1_asimow_roth}. The $|\mathcal{E}| \times 2n$ rigidity matrix $R(p) := \frac{\partial F(p)}{\partial p}$ is the Jacobian of $F$ \cite{gortler2025higher}. Let $p^*$ denote a target equilibrium and assume the configuration $p^*$ affinely spans $\mathbb{R}^2$. Then $\ker(R(p^*))$ contains at least a $3$-dimensional subspace of trivial \ac{rbms}. When this is exact, i.e., when $\operatorname{dim}\ker(R(p^*))=3$, the framework is \textit{infinitesimally rigid}; if $\dim\ker(R(p^*)) = 3 + n_f$ with $n_f > 0$, it is \textit{flexible}, admitting $n_f$ internal flexes $\{z_k\}_{k=1}^{n_f}$ \cite{connelly1980rigidity}.

The \ac{rbm} subspace is exactly three-dimensional, consisting of two uniform translations and one global rotation. Their local $2$-dimensional spatial components at agent $k$ are explicitly given by
\begin{align} 
    &[v_x]_k = \frac{1}{\sqrt{n}} \begin{bmatrix} 1 \\ 0 \end{bmatrix}, \quad [v_y]_k = \frac{1}{\sqrt{n}} \begin{bmatrix} 0 \\ 1 \end{bmatrix}, \label{eq:rbm_trans} \\ &[v_r]_k = \frac{1}{\sqrt{J_p}} \Omega (p_k^* - p_{\mathrm{cm}}), \label{eq:rbm_rot}
\end{align}
where $p_{\mathrm{cm}} = \frac{1}{n}\sum_{k=1}^n p_k^*$ is the center of mass and $J_p = \sum_{k=1}^n \|p_k^* - p_{\mathrm{cm}}\|^2$ is the polar moment of inertia.

\subsection{Input-Output Dynamics}
\label{subsec:IO-dynamics}
To maintain the target formation $p^*$, the agents employ a distributed gradient-descent control law derived from the potential $V(p) = \frac{1}{2}\|F(p) - F(p^*)\|^2$ \cite{survey_formation_control_2015_automatica, krick2009stabilisation, OLFATISABER2002495}. We consider a single-actuator, single-sensor scenario where a disturbance $w \in \mathbb{R}^2$ enters at an actuated agent $i$, and the displacement $y \in \mathbb{R}^2$ is measured at a sensor-equipped agent $j$. The closed-loop input-output model is
\begin{align}
    \dot{p} &= -R(p)^\top\big(F(p) - F(p^*)\big) + Bw, \label{eq:nonlinear_io} \\
    y &= Cp, \label{eq:nonlinear_out}
\end{align}
where $B = e_i \otimes I_2$ and $C = e_j^\top \otimes I_2$.

To characterize the structural limitations of the network using classical input-output systems theory, we analyze the local behavior around the target equilibrium $(p^*, w=0)$. Letting $\delta p = p - p^*$ be the state variation, the linearized LTI dynamics are governed by the configuration-dependent stiffness matrix:
\begin{equation} \label{eq:lti_model}
    \delta\dot{p} = -R(p^*)^\top R(p^*) \delta p + B w, \qquad y = C \delta p.
\end{equation}
The stiffness matrix $A := -R(p^*)^\top R(p^*)$ inherits the spectral structure of $R(p^*)$; its kernel modes correspond to the persistent \ac{rbms} in addition to any flexes if they are present, while its strictly negative eigenvalues govern the transient, shape-changing dynamics.

Under this formulation, our objective is to identify the precise geometric configurations under which the steady-state input-output map from an actuated agent $i$ to a sensor-equipped agent $j$ becomes rank-deficient, and to derive spatial allocation conditions that prevent such steady-state signal blocking.

\section{Steady-State Analysis and\\ DC Gain Rank Drop}
\label{sec:steady_state_analysis}
This section establishes the exact algebraic and geometric conditions that induce a steady-state transmission zero. We first derive the multivariable DC gain matrix from the network's modal dynamics. We then demonstrate how the configuration-dependent cross-coupling required for signal blocking in flexible networks collapses into explicit, predictable spatial boundaries for infinitesimally rigid frameworks.

\subsection{The Transfer Function and Modal Decomposition}
To characterize how a disturbance applied at agent $i$ propagates through the network and appears at agent $j$, we analyze the input--output transfer operator of the linearized system \eqref{eq:lti_model}. Passing to the frequency domain yields $Y(s) = G_{ji}(s) W(s)$, where $G_{ji}(s) \in \mathbb{C}^{2 \times 2}$ denotes the transfer matrix mapping a $2$-dimensional disturbance applied at agent $i$ to the $2$-dimensional displacement measured at agent $j$. The transfer matrix can be written compactly as:
\begin{equation}
    G_{ji}(s) = C(sI + R(p^*)^\top R(p^*))^{-1} B.
\end{equation}

Utilizing the spectral decomposition of the symmetric stiffness matrix $-R(p^*)^\top R(p^*) = \sum_{k=1}^{2n} \lambda_k \phi_k \phi_k\!^\top$, we express the transfer function in its modal form. This naturally separates the stable, shape-changing deformational modes ($\lambda_k < 0$) from the non-decaying modes ($\lambda_k = 0$):
\begin{equation} \label{eq:transfer_function_modal}
    G_{ji}(s) = \frac{[P_0]_{ji}}{s} + \sum_{k: \lambda_k < 0} \frac{[\phi_k]_j [\phi_k]_i^\top}{s - \lambda_k},
\end{equation}
where $[\phi_k]_i := (e_i^\top \otimes I_2)\phi_k \in \mathbb{R}^2$ denotes the local $2$-dimensional spatial component of eigenvector $\phi_k$ at agent $i$. 

The matrix $[P_0]_{ji}$ is the $(j, i)$-th $2 \times 2$ block of the orthogonal projection matrix $P_0$ onto $\ker(R(p^*))$. This decomposition reveals that high-frequency disturbances are filtered through the network's deformational modes, while the steady-state behavior is governed entirely by the DC gain matrix $[P_0]_{ji}$.

For an infinitesimally rigid framework, where internal flexes are absent, the null space consists strictly of the \ac{rbms}. The orthogonal projection matrix onto this purely rigid subspace, denoted $P_{\mathrm{rbm}}$, is constructed from the basis modes defined in \eqref{eq:rbm_trans} and \eqref{eq:rbm_rot}:
\begin{align} \label{eq:P0_rigid_explicit}
    [P_{\mathrm{rbm}}]_{ji} &= [v_x]_j [v_x]_i^\top + [v_y]_j [v_y]_i^\top + [v_r]_j [v_r]_i^\top  \notag \\ &=\frac{1}{n}I_2 + [v_r]_j [v_r]_i^\top.
\end{align}

For flexible frameworks ($n_f > 0$), letting $\{z_1, \dots, z_{n_f}\}$ denote an orthonormal basis for the internal flexes, the total projection matrix must account for these additional dimensions:
\begin{equation}
    [P_0]_{ji} = [P_{\mathrm{rbm}}]_{ji} + \sum_{k=1}^{n_f} [z_k]_j [z_k]_i^\top.
\end{equation}

\subsection{Steady-State Transmission Zeros}
Following \eqref{eq:transfer_function_modal}, we focus on steady-state input--output blocking. Since the deformational terms are strictly stable, their contribution vanishes as $t\to\infty$, and the steady-state mapping is governed by the kernel term. We say that the input--output channel from agent $i$ to agent $j$ exhibits a \emph{steady-state transmission zero} if the DC gain matrix is rank-deficient, i.e., 
\begin{equation}
    \mathrm{rank} \left( \lim_{s \to 0} s G_{ji}(s) \right) < 2 \iff \det([P_0]_{ji}) = 0,
\end{equation}
which corresponds to the existence of a nonzero input direction $u\in\mathbb{R}^2$ satisfying $[P_0]_{ji}u=0$.

To evaluate this determinant algebraically, we recall that the two translational \ac{rbms} uniformly contribute a baseline of $\frac{1}{n}I_2$ to the projection. Thus, we can express the DC gain matrix generally as $[P_0]_{ji} = \frac{1}{n}I_2 + U V^\top$, where the matrices $U, V \in \mathbb{R}^{2 \times m}$ (with $m = 1 + n_f$) collect the local components of the rotational \ac{rbm} and internal flex modes at agents $j$ and $i$, respectively.

Applying Sylvester's determinant identity, the determinant evaluates to
\begin{equation} \label{eq:sylvester}
    \det([P_0]_{ji}) = \det\left(\frac{1}{n}I_2\right) \det(H(p)),
\end{equation}
where the $m \times m$ modal cross-coupling matrix $H(p) \coloneqq I_m + n V^\top U$ partitions into rotational and flex components,
\begin{equation} \label{eq:block_matrix}
    H(p) = \begin{bmatrix} 1 + n\Psi_{rr} & n\Psi_{rz} \\ n\Psi_{zr} & I_{n_f} + n\Psi_{zz} \end{bmatrix},
\end{equation}
where the submatrices quantify cross-agent modal overlaps: $\Psi_{rr} = \langle [v_r]_i, [v_r]_j \rangle$, $[\Psi_{zz}]_{ab} = \langle [z_a]_i, [z_b]_j \rangle$, $[\Psi_{rz}]_{1b} = \langle [v_r]_i, [z_b]_j \rangle$, and $[\Psi_{zr}]_{a1} = \langle [z_a]_i, [v_r]_j \rangle$. Because the local spatial geometries at agents $i$ and $j$ differ, these cross-couplings are not generally symmetric. 

Provided the actuator and sensor are not pinned \cite{casspi2025geometry}, a DC transmission zero ($\det(H(p))=0$) strictly requires the configuration-dependent modal contributions to perfectly destructively interfere with the uniform translational baseline.


\subsection{Generic Configurations and Real Algebraic Sets}

\begin{proposition}[Measure-Zero Rarity] \label{prop:measure_zero}
    For a generic connected, flexible planar framework, the set of sensor-actuator embeddings that induce a steady-state transmission zero forms a proper real algebraic variety of Lebesgue measure zero in the configuration space.
\end{proposition}

\begin{proof}
Let $p \in \mathbb{R}^{2n}$ denote the spatial configuration. By Sylvester's determinant identity used in \eqref{eq:sylvester}, the steady-state rank-loss condition is equivalent to $\det(H(p)) = 0$.

The rigidity matrix $R(p) \in \mathbb{R}[p]^{|\mathcal{E}| \times 2n}$ is a polynomial matrix since it is the Jacobian of squared inter-agent distances. For a connected, flexible planar framework, the generic rank of $R(p)$ is $r = 2n - (3+n_f)$, so the generic null-space dimension is $(3+n_f)$. 
By permuting columns and restricting to configurations where a chosen \(r\times r\) minor is nonzero, we partition
\[
R(p) =
\begin{bmatrix}
R_{\mathrm{piv}}(p) & R_{\mathrm{free}}(p)
\end{bmatrix},
\]
where the pivot block \(R_{\mathrm{piv}}(p)\in\mathbb{R}^{r\times r}\) is nonsingular.



To construct a basis matrix for $\ker(R(p))$, we solve the right-hand-side system
\[
R_{\mathrm{piv}}(p) Z_{\mathrm{piv}} = -R_{\mathrm{free}}(p) Z_{\mathrm{free}} .
\]
Setting the free block to the identity matrix $Z_{\mathrm{free}} = I_{3+n_f}$ generates $3+n_f$ independent nullspace vectors. Using the identity
\[
R_{\mathrm{piv}}^{-1}(p) = \frac{\operatorname{adj}(R_{\mathrm{piv}}(p))}{\det(R_{\mathrm{piv}}(p))},
\]
we obtain
\begin{equation}\label{Zbasis}
Z_{\mathrm{piv}} =
- R_{\mathrm{piv}}^{-1}(p) R_{\mathrm{free}}(p)
=
-\frac{\operatorname{adj}(R_{\mathrm{piv}}(p))}{\det(R_{\mathrm{piv}}(p))} R_{\mathrm{free}}(p).
\end{equation}
Multiplying the assembled basis by the scalar polynomial $\det(R_{\mathrm{piv}}(p))$ clears the denominator without changing the column span on the generic configuration set, yielding the polynomial nullspace basis
\[
K(p) =
\begin{bmatrix}
-\operatorname{adj}(R_{\mathrm{piv}}(p)) R_{\mathrm{free}}(p) \\
\det(R_{\mathrm{piv}}(p)) I_{3+n_f}
\end{bmatrix}
\in \mathbb{R}[p]^{2n \times (3+n_f)}.
\]

To isolate the internal flexes $Z(p)$, we project $K(p)$ onto the orthogonal complement of the rigid-body motions,
\[
Z(p) = \big(I_{2n} - P_{\mathrm{rbm}}(p)\big) K(p).
\]
Since $K(p)$ has polynomial entries and the rigid-body projector $P_{\mathrm{rbm}}(p)$ is a rational matrix in the configuration coordinates $p$, the resulting matrix $Z(p)$ has entries in the field of rational functions $\mathbb{R}(p)$.

Consequently, each entry of the matrix \(H(p)\) is a rational function of the configuration coordinates \(p\) on the generic domain under consideration. Therefore its determinant is also rational, and may be written as
\[
\det(H(p)) = \frac{N(p)}{D(p)},
\]
for some polynomials \(N,D \in \mathbb{R}[p]\), with \(D \not\equiv 0\). The polynomial \(D(p)\) is obtained by clearing the denominators arising from the rigid-body projector \(P_{\mathrm{rbm}}(p)\) and from the inverse of the pivot block \(R_{\mathrm{piv}}(p)\) used in \eqref{Zbasis}. On the generic domain, these denominators do not vanish, since the selected pivot minor satisfies \(\det(R_{\mathrm{piv}}(p)) \neq 0\) and the rigid-body projector is well defined. Hence, on this domain $\det(H(p)) = 0$ if and only if $N(p)=0$.


To complete the argument, we show that the numerator polynomial \(N\) is not identically zero. From \eqref{eq:sylvester}, the condition \(N(p)=0\) is equivalent to $\det([P_0(p)]_{ji}) = 0$.  Consider any generic configuration \(\bar p\) for which the actuator and sensor are not pinned and the corresponding modal components at agents \(i\) and \(j\) are linearly independent. For such a configuration the matrix \([P_0(\bar p)]_{ji}\) has full rank, so
\[
\det([P_0(\bar p)]_{ji}) \neq 0 .
\]
Therefore \(N(\bar p)\neq 0\), implying that \(N\) is not the zero polynomial.


Thus, on the generic domain under consideration, the configurations satisfying the rank-loss condition coincide with the real algebraic set
\[
\mathcal{Z}(N) \coloneqq \{p \in \mathbb{R}^{2n} \mid N(p) = 0\}.
\]
Since \(N\) is not the zero polynomial, \(\mathcal{Z}(N)\) is a proper real algebraic subset of \(\mathbb{R}^{2n}\) \cite{bochnak1998real}. It is well known that the zero set of a nontrivial polynomial in \(\mathbb{R}^{2n}\) has Lebesgue measure zero \cite{lee2012introduction}.
\end{proof}

Practically, Proposition \ref{prop:measure_zero} implies that deploying a flexible network with a generic (random) geometric embedding yields a zero probability of encountering a steady-state structural zero. While these exact cancellations are strictly confined to lower-dimensional algebraic sets, their internal mechanics provide critical physical insight. 
Let \(M := I_{n_f}+n\Psi_{zz}\). \textcolor{black}{For configurations where $M$ is invertible, applying the Schur complement to $H(p)$ allows the determinant condition \eqref{eq:sylvester} to be evaluated as:}
\begin{equation} \label{eq:schur_complement}
\det(M)\big(1+n\Psi_{rr}-n^2\Psi_{rz}M^{-1}\Psi_{zr}\big)=0 .
\end{equation}
%
Since the determinant term is non-zero in these cases, the scalar bracket must vanish. This reveals that a structural zero requires the pure rigid-body transmission ($1 + n\Psi_{rr}$) to be perfectly annihilated by the cross-coupling with the internal deformational modes. 

The sheer algebraic complexity of maintaining this exact balance across multiple highly coupled modes intuitively reinforces their measure-zero rarity. Furthermore, because internal flexes generally lack closed-form spatial representations, mapping the algebraic zero-set $\mathcal{Z}(N)$ into explicit, actionable geometric boundaries for sensor placement is analytically intractable. This geometric fragility, contrasted against the need for rigorous design rules, directly motivates our pivot to the infinitesimally rigid baseline, a structural transition explicitly visualized in Section \ref{sec:simulations}.

\subsection{The Spatial Locus of Transmission Zeros}
Consider now an infinitesimally rigid framework, where internal flexes are identically absent ($n_f = 0$). Consequently, the cross-coupling terms between flexes and rotations vanish, and the general Schur complement condition in \eqref{eq:schur_complement} collapses cleanly to the simple scalar equation,
\begin{equation} \label{eq:rigid_scalar}
    1 + n \langle [v_r]_i, [v_r]_j \rangle = 0.
\end{equation}
Substituting the explicit geometric form of this rigid-body motion yields a linear constraint in the spatial coordinates of the sensor. This motivates the following geometric definition.

\begin{definition}[Spatial Locus of Transmission Zeros] Consider an infinitesimally rigid formation $(\mathcal{G}, p^*)$ in $\mathbb{R}^2$, and fix an actuated agent $i$. The \emph{spatial locus of transmission-zero} associated with agent $i$ is the set
\begin{equation}\label{eq:locus_def}
\mathcal{L}_i
:=
\left\{
x \in \mathbb{R}^2
\;\middle|\;
\langle p_i^* - p_{\mathrm{cm}}, x - p_{\mathrm{cm}} \rangle
=
-\frac{J_p}{n}
\right\}.
\end{equation}
\end{definition}

Geometrically, the set $\mathcal{L}_i$ defines an affine hyperplane orthogonal to the relative position vector $p_i^* - p_{\mathrm{cm}}$ and located in the half-plane opposite to the actuated agent. This leads to our main geometric result for formations in $\mathbb{R}^2$.

\begin{theorem}\label{thm:spatial-locus}
For an infinitesimally rigid formation in $\mathbb{R}^2$ actuated at agent $i$,
the DC gain matrix $[P_0]_{ji}$ loses rank if and only if $p_j^* \in \mathcal{L}_i$.
\end{theorem}


\begin{proof}
Substituting the explicit physical representation of the normalized global rotational mode from \eqref{eq:rbm_rot} into the scalar rank-drop condition \eqref{eq:rigid_scalar} yields
\begin{equation*}
    1 + \frac{n}{J_p} \langle \Omega (p_i^* - p_{\mathrm{cm}}), \Omega (p_j^* - p_{\mathrm{cm}}) \rangle = 0.
\end{equation*}
Because the 2D rotational generator $\Omega$ is an orthogonal matrix ($\Omega^\top \Omega = I_2$), the inner product is invariant under rotation. The expression simplifies to the standard Euclidean inner product of the relative position vectors,
\begin{equation}\label{eq:spatial_locus}
    1 + \frac{n}{J_p} \langle p_i^* - p_{\mathrm{cm}}, p_j^* - p_{\mathrm{cm}} \rangle = 0.
\end{equation}
Rearranging this equation directly isolates the spatial constraints on the measured agent $j$,
\begin{equation*}
    \langle p_i^* - p_{\mathrm{cm}}, p_j^* - p_{\mathrm{cm}} \rangle = -\frac{J_p}{n},
\end{equation*}
which exactly matches the geometric affine hyperplane defining the spatial locus.
\end{proof}

\begin{remark}[Directional Signal Blocking]
Since the DC gain matrix $[P_0]_{ji} \in \mathbb{R}^{2 \times 2}$ has a strictly positive trace (due to the $\frac{1}{n}I_2$ translation baseline), the condition $\det([P_0]_{ji}) = 0$ implies a rank drop from 2 to 1, never to 0. Physically, a sensor on this locus is not completely blind but suffers severe directional signal blocking, unable to detect disturbances projecting onto the one-dimensional null space of the DC gain matrix. \hfill $\Diamond$
\end{remark}

Geometrically, this transmission zero boundary is governed entirely by the formation's geometry. The polar moment of inertia $J_p$ acts as a spatial scaling parameter. For a fixed relative actuator position, formations with a larger spatial spread (larger $J_p$) push the zero-transmission hyperplane further away from the center of mass, thereby expanding the permissible regions for sensor placement. Conversely, highly compact formations draw these zero-transmission hyperplanes proportionally closer to the center of mass.

\begin{remark}[Contrast with State-Space Uncontrollability]
The spatial locus complements the geometric controllability boundaries in \cite{casspi2025geometry}. There, uncontrollability arises because the actuator lies at the center of a pure rotation, making its local component strictly zero and the mode simply cannot be excited. Here, the phenomenon is fundamentally different: all rigid-body modes are fully excited and the sensor is not pinned, yet the uniform translational contribution to the DC gain and the position-dependent rotational contribution cancel exactly along a spatial hyperplane. Thus, while the actuator's position dictates where rotations cannot be excited (state-space uncontrollability), the global geometric invariants $p_{\mathrm{cm}}$ and $J_p$, dictates where fully-excited motions become invisible at the output. \hfill $\Diamond$
\end{remark}

From a design perspective, this theorem provides a rigorous geometric constraint: to ensure full steady-state transmission of a localized disturbance at agent $i$, the sensor $j$ must not be placed anywhere along this hyperplane.

\begin{corollary}[Actuator-Sensor Reciprocity]
For an infinitesimally rigid formation, the geometric condition for a steady-state transmission zero is strictly reciprocal. That is, if the steady-state mapping from an actuated agent $i$ to a measured agent $j$ is rank-deficient, then the reciprocal mapping from agent $j$ to agent $i$ is identically rank-deficient.
\end{corollary}

\begin{proof}
The scalar rank-drop condition \eqref{eq:spatial_locus} requires $\langle p_i^* - p_{\mathrm{cm}}, p_j^* - p_{\mathrm{cm}} \rangle = -J_p / n$. By the commutativity of the Euclidean inner product, this condition is perfectly invariant under the permutation of the indices $i$ and $j$.
\end{proof}

\subsection{The Global Transmission Polygon}
Theorem \ref{thm:spatial-locus} provides a strict sensor placement constraint given a known actuator location. However, in practical multi-agent deployments, the network may be subjected to exogenous disturbances at unknown or arbitrary nodes. To guarantee steady-state signal transmission against all possible single-node excitations, the measured agent must avoid the spatial loci of all potential actuators simultaneously.

For an actuated agent $i \in \mathcal{V}$, the zero-transmission hyperplane $\mathcal{L}_i$ partitions $\mathbb{R}^2$ into two half-spaces. We define the strictly safe transmission half-space $\mathcal{H}_i$ containing the center of mass as,
\begin{equation}
    \mathcal{H}_i = \left\{ x \in \mathbb{R}^2 \;\middle|\; \langle p_i^* - p_{\mathrm{cm}}, x - p_{\mathrm{cm}} \rangle > -\frac{J_p}{n} \right\}.
\end{equation}

To ensure full rank of the DC gain matrix $[P_0]_{ji}$ for a disturbance at any arbitrary node, the sensor must reside within the intersection of all such safe half-spaces. We define this region below.
\begin{definition}[Global Transmission Polygon]
For a planar infinitesimally rigid formation, the \emph{global transmission polygon} is the convex set
\begin{equation}
\mathcal{C} := \bigcap_{i \in \mathcal{V}} \mathcal{H}_i.
\end{equation}
\end{definition}

Furthermore, $\mathcal{C}$ is strictly non-empty because it contains the center of mass; evaluating $x=p_{\mathrm{cm}}$ yields $0>-J_p/n$, which holds intrinsically since $J_p,n>0$.

\begin{proposition}[Omnidirectional Transmission]
If a sensor is placed on an agent $j$ such that $p^*_j\in \mathcal{C}$, the network's steady-state DC gain matrix $[P_0]_{ji}$ retains full rank for any actuated agent $i\in \mathcal{V}$.
\end{proposition}

\begin{proof}
If $p^*_j \in \mathcal{C}$, then by definition $p^*_j \in \mathcal{H}_i$ for all $i\in \mathcal{V}$. Consequently, the sensor does not intersect any zero-transmission hyperplane $\mathcal{L}_i$, and by Theorem \ref{thm:spatial-locus}, no steady-state structural transmission zeros can exist.

\end{proof}

From a network synthesis perspective, the global transmission polygon dictates that sensors should be allocated to agents located nearest to the formation's center of mass. Agents on the geometric periphery are more likely to fall outside this polygon, rendering them susceptible to steady-state transmission zeros from disturbances originating on the opposite side of the formation.

\begin{corollary}[Collocated Transmission]
For an infinitesimally rigid formation in $\mathbb{R}^2$, a collocated actuator-sensor pair cannot produce a steady-state transmission zero. In particular, $p_i^* \notin \mathcal{L}_i$ for all $i \in \mathcal{V}$.
\end{corollary}

\begin{proof}
Consider the collocated case where $i = j$. The rank-drop condition \eqref{eq:spatial_locus} evaluates to $\|p_i^* - p_{\mathrm{cm}}\|^2 = -J_p / n$. Since the left-hand side is non-negative and the right-hand side is strictly negative ($J_p > 0, n > 0$), the equality cannot be satisfied for any real framework embedding. Thus, the condition for a steady-state transmission zero is never met, ensuring $p_i^* \in \mathcal{H}_i$.
\end{proof}

Concluding this section, we have established a geometric-driven analysis for steady-state transmission zeros. We proved that connected, flexible networks only admit structural zeros at non-generic embeddings that occupy measure-zero subsets of the configuration space. For the strictly infinitesimally rigid case, we derived an explicit geometric condition that characterizes the spatial locus of these zeros as affine hyperplanes. By synthesizing these boundaries into the global transmission polygon, we provide a geometry-based design rule for sensor allocation. We now validate these exact theoretical limits and demonstrate the structural transition between flexible and rigid embeddings through numerical simulations.

\section{Numerical Simulations}
\label{sec:simulations}

To validate our geometric framework, we construct an asymmetric 4-agent formation in $\mathbb{R}^2$ with $p_{\mathrm{cm}} = (0,0)$ and polar moment $J_p = 16$. The agents are positioned at $p_1^*=(1.5, 1)$, $p_2^*=(-2, -1)$, $p_3^*=(1.5, -1.5)$, and $p_4^*=(-1, 1.5)$. To construct the formation's global transmission polygon, Figure \ref{fig:transmission_core_exact} plots the spatial loci $\mathcal{L}_i$ for each agent $i \in \{1,2,3,4\}$, representing the zero-transmission hyperplanes if any agent were to act as the exogenous actuator. To explicitly demonstrate the resulting steady-state signal blocking, we designate agent 1 as the actuator. By allocating the sensor to agent 2, whose target coordinate lands exactly on the zero-transmission hyperplane $\mathcal{L}_1$, we deliberately enforce the geometric rank-drop condition. 

Figure \ref{fig:singular_values} verifies this structural blocking: the minimum singular value of $sG_{ji}(j\omega)$ tends to zero as $\omega \to 0$, confirming a rank drop of the DC gain, validating Theorem \ref{thm:spatial-locus}.

Finally, to demonstrate the complexity of the general flexible condition \eqref{eq:schur_complement}, we remove the internal bracing edge (dashed black line, Fig. \ref{fig:transmission_core_exact}) to create a flexible 4-bar linkage. As shown in Figure \ref{fig:singular_values}, the introduction of internal flexes fundamentally alters the required geometric cross-coupling, breaking the precise destructive interference and restoring the full DC rank of the transfer matrix.

\begin{figure}[htbp!]
    \centering
    \begin{tikzpicture}
        \begin{axis}[
            width=2.9in,               
            axis equal image,
            grid=both,
            grid style={dashed, gray!30},
            xlabel={Position $x_1$},
            ylabel={Position $x_2$},
            label style={font=\small},
            tick label style={font=\footnotesize},
            xtick={-4, -3, -2, -1, 0, 1, 2, 3, 4, 5},
            ytick={-4, -3, -2, -1, 0, 1, 2, 3, 4},
            xmin=-4.5, xmax=5.5,
            ymin=-4.5, ymax=4.5
        ]

        \tikzstyle{agent}=[circle, fill=black, inner sep=0pt, minimum size=5pt]
        \tikzstyle{sensor}=[circle, fill=cyan!80!blue, inner sep=0pt, minimum size=5pt]
        \tikzstyle{edge}=[thick, gray!80]
        \tikzstyle{flexedge}=[thick, dashed, black]
        \tikzstyle{locus}=[red, thick, dashed]


        \fill[green!20, opacity=0.7] 
            (2.5, -1.0) --
            (0.44, 3.11) --
            (-2.67, 0.0) --
            (-0.62, -3.08) -- cycle; 

        \draw[locus] (-4.0, 2.0) -- (0.33, -4.5) node[pos=0.15, above right, font=\footnotesize, text=red, fill=white, inner sep=1pt] {$\mathcal{L}_1$};
        \draw[locus] (0.0, 4.0) -- (4.25, -4.5) node[pos=0.80, above right, font=\footnotesize, text=red, fill=white, inner sep=1pt] {$\mathcal{L}_2$};
        \draw[locus] (-4.5, -1.83) -- (1.83, 4.5) node[pos=0.85, above left, font=\footnotesize, text=red, fill=white, inner sep=1pt] {$\mathcal{L}_3$};
        \draw[locus] (-2.75, -4.5) -- (5.5, 1.0) node[pos=0.90, below, font=\footnotesize, text=red, fill=white, inner sep=1pt] {$\mathcal{L}_4$};

        \draw[edge] (-2.0, -1.0) -- (1.5, -1.5);     
        \draw[edge] (1.5, 1.0) -- (1.5, -1.5);       
        \draw[edge] (-2.0, -1.0) -- (-1.0, 1.5);    
        \draw[edge] (-1.0, 1.5) -- (1.5, 1.0);      
        \draw[flexedge] (1.5, -1.5) -- (-1.0, 1.5); 

        \node[agent] at (1.5, 1.0) {};       
        \node[sensor] at (-2.0, -1.0) {};    
        \node[agent] at (1.5, -1.5) {};      
        \node[agent] at (-1.0, 1.5) {};      
        \node[black, above right, font=\tiny, inner sep=1pt] at (1.5, 1.0) {$p_1^*$};
        \node[blue, below left, font=\tiny, inner sep=1pt] at (-2.0, -1.0) {$p_2^* \in \mathcal{L}_1$};
        \node[black, below right, font=\tiny, inner sep=1pt] at (1.5, -1.5) {$p_3^*$};
        \node[black, above left, font=\tiny, inner sep=1pt] at (-1.0, 1.5) {$p_4^*$};

        \node[red, font=\small] at (0,0) {$\times$};
        \node[red, below right, font=\footnotesize, fill=white, fill opacity=0.7, inner sep=1pt] at (0,0) {$p_{\mathrm{cm}}$};
        \end{axis}
    \end{tikzpicture}
    \caption{Global Transmission Polygon for an asymmetric 4-agent formation. Agent $p_2^*$ lies precisely on $\mathcal{L}_1$, inducing an exact rank loss ($\langle p_1^* - p_{\mathrm{cm}}, p_2^* - p_{\mathrm{cm}} \rangle = -J_p/n = -4$). By reciprocity, $p_1^* \in \mathcal{L}_2$. Agents $p_3$ and $p_4$ reside strictly inside the shaded polygon, guaranteeing omnidirectional signal transmission.}
    \label{fig:transmission_core_exact}
\end{figure}
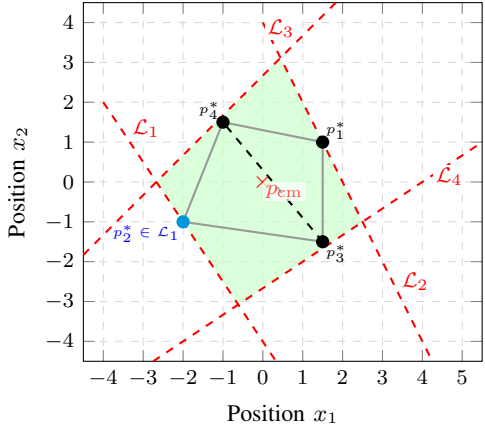

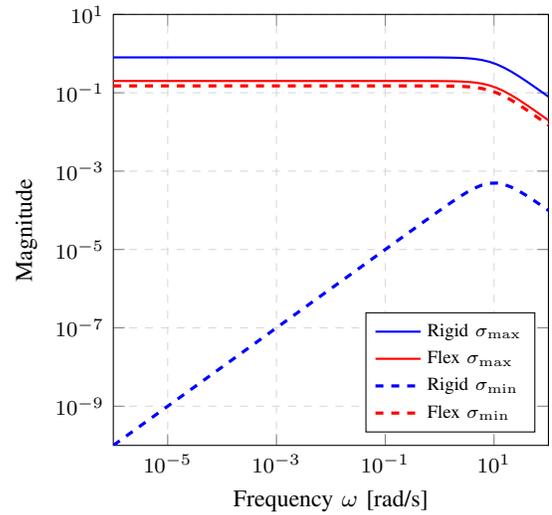
\begin{figure}[htbp!]
    \centering
    \begin{tikzpicture}
        \begin{loglogaxis}[
            width=2.9in,
            height=2.88in,             
            grid=both,
            grid style={dashed, gray!30},
            xlabel={Frequency $\omega$ [rad/s]},
            ylabel={Magnitude},
            xmin=1e-6, xmax=1e2,
            ymin=1e-10, ymax=10,       
            label style={font=\small},
            tick label style={font=\footnotesize},
            legend style={font=\scriptsize, at={(0.97,0.03)}, anchor=south east},
            legend cell align={left}
        ]

        \addplot[thick, blue, domain=1e-6:1e2, samples=60] {0.8 / sqrt(1 + (x/10)^2)};
        \addlegendentry{Rigid $\sigma_{\max}$}

        \addplot[thick, red, domain=1e-6:1e2, samples=60] {0.2 / sqrt(1 + (x/10)^2)};
        \addlegendentry{Flex $\sigma_{\max}$}

        \addplot[very thick, dashed, blue, domain=1e-6:1e2, samples=60] {1e-4 * x / (1 + (x/10)^2)};
        \addlegendentry{Rigid $\sigma_{\min}$}

        \addplot[very thick, dashed, red, domain=1e-6:1e2, samples=60] {0.15 / sqrt(1 + (x/10)^2)};
        \addlegendentry{Flex $\sigma_{\min}$}

        \end{loglogaxis}
    \end{tikzpicture}
    \caption{Singular values of $sG_{ji}(j\omega)$ for the 4-agent formation. The maximum singular values (solid lines) remain strictly positive due to the invariant uniform translational modes. For the infinitesimally rigid framework, the minimum singular value (blue dashed) strictly drops as $\omega \to 0$, confirming a DC rank of 1. Introducing internal flexes breaks this exact geometric cancellation, restoring the full DC rank (red dashed).}
    \label{fig:singular_values}
\end{figure}


\section{Conclusion} \label{sec:conclusion}

This letter presented a geometric framework characterizing steady-state transmission zeros in distance-based formations. We proved that structural zeros in flexible planar frameworks are confined to measure-zero algebraic sets, while infinitesimally rigid formations natively resolve these conditions into explicit spatial hyperplanes. These boundaries form the global transmission polygon, yielding a direct synthesis rule for robust sensor-actuator placement.

Extending to arbitrary dimensions $d\geq 3$ requires generalizing the algebraic rank-drop conditions and navigating the non-commutative Lie algebra of $SO(d)$ to extract explicit geometric boundaries.

Furthermore, future work will transition from this linearized input-output analysis to the exact nonlinear dynamics, investigating the geometric configurations that induce nonlinear indistinguishability to bridge our steady-state structural zero analysis with the system's exact observability properties.

\bibliographystyle{IEEEtran}
\bibliography{bibio}

\end{document}